\documentclass[11pt]{article}

\usepackage{amsmath,amssymb,amsthm,amsfonts,amstext}
\usepackage{graphicx,algorithmic}
\usepackage{times,epsfig}
\usepackage{enumerate}
\usepackage[letterpaper,margin=1in]{geometry}

\newtheorem{thm}{Theorem}
\newtheorem{lemma}[thm]{Lemma}

\newtheorem{obs}[thm]{Observation}
\newtheorem{claim}{Claim}

\theoremstyle{remark}

\theoremstyle{definition}

\newtheorem{algorithm}[thm]{Algorithm}

\newcommand{\grp}[1]{\left(#1\right)}
\newcommand{\set}[1]{\left\{#1\right\}}
\newcommand{\ignore}[1]{}
\newcommand{\dplus}{\ensuremath{\Delta^+}} 
\newcommand{\dcrk}{\ensuremath{(\mbox{DCR}(k)})}
\newcommand{\0}{\textbf{0}}
\newcommand{\1}{\textbf{1}}

\newcommand{\bcr}{\eqref{bcr}}
\newcommand{\dcr}{\eqref{dcr}}
\newcommand{\bcrf}{\eqref{bcrf}}
\newcommand{\dcrf}{\eqref{dcrf}}
\newcommand{\BCR}[1]{(BCR$_{#1}$)}
\newcommand{\DCR}[1]{(DCR$_{#1}$)}

\newcommand{\ddcr}[1]{(DCR$_{#1}^D)$}

\def\C{\ensuremath{\mathcal{C}}}

\def\I{\ensuremath{\mathcal{I}}}
\def\K{\ensuremath{\mathcal{K}}}
\def\L{\ensuremath{\mathcal{L}}}

\def\X{\ensuremath{\mathcal{X}}}
\def\Y{\ensuremath{\mathcal{Y}}}

\def\R{\mathbb{R}}

\def \naturals {\mathbb{N}}

\def\sink{\ensuremath{\operatorname{sink}}}
\def\sources{\operatorname{sources}}

\def\max{\operatorname{max}}

\def\st{\textup{s.t.}}
\def\supp{\operatorname{supp}}

\def\comp{\mathcal{K}}

\title{Efficient Algorithms for Solving Hypergraphic Steiner Tree
  Relaxations in Quasi-Bipartite Instances}
\author{
Isaac Fung \and Konstantinos Georgiou \and Jochen K\"onemann \and Malcolm Sharpe }

\begin{document}
\maketitle

\abstract{We consider the Steiner tree problem in quasi-bipartite
  graphs, where no two Steiner vertices are connected by an edge. For
  this class of instances, we present an efficient algorithm to
  exactly solve the so called directed component relaxation (DCR), a
  specific form of hypergraphic LP relaxation that was instrumental in
  the recent break-through result by Byrka et al.~\cite{BGRS10}. Our
  algorithm hinges on an efficiently computable map from extreme
  points of the bidirected cut relaxation to feasible solutions of
  (DCR). As a consequence, together with \cite{BGRS10} we immediately
  obtain an efficient 73/60-approximation for quasi-bipartite Steiner
  tree instances. We also present a particularly simple (BCR)-based
  random sampling algorithm that achieves a performance guarantee
  slightly better than 77/60. }

\section{Introduction}

In the Steiner tree problem, we are given an undirected graph
$G=(V,E)$ with costs $c$ on edges and its vertex set partitioned into
terminals (denoted $R \subseteq V$) and Steiner vertices ($V\setminus
R$). A \emph{Steiner tree} is a tree spanning all of $R$ plus any
subset of $V \subseteq R$, and the problem is to find a minimum-cost such
tree.  The Steiner tree problem is {$\mathsf{APX}$}-hard, thus the
best we can hope for is a constant-factor approximation algorithm. In particular, the best inapproximability known (assuming $\mathsf{P} \not = \mathsf{NP}$) is 1.01063 ($>\frac{96}{95}$) due to Chleb{\'i}k and Chleb{\'i}kov{\'a}~\cite{CC02}. For the special family of instances that are known as quasi-bipartite graphs, which is the subject of our work, the best hardness known is 1.00791 ($>\frac{128}{127}$)~\cite{CC02}, under the same complexity assumption. 

In a recent break-through paper, Byrka, Grandoni, Rothvo{\ss} and
Sanit\`{a} \cite{BGRS10,BGRS11} presented the currently best approximation algorithm known for the problem. The algorithm has a
performance ratio of $\ln(4)+\epsilon$ for any fixed $\epsilon > 0$,
and it iteratively rounds solutions to a so called {\em hypergraphic}
linear program. Such LPs commonly have a variable $x_K$ for each 
$K \subset R$, representing a {\em full component} spanning the
terminals of $K$. A full component is a tree whose leaves are
terminals and whose internal vertices are non-terminals.

\begin{figure}\label{fig:fullcomponents}
\begin{center}
\includegraphics[scale=0.7]{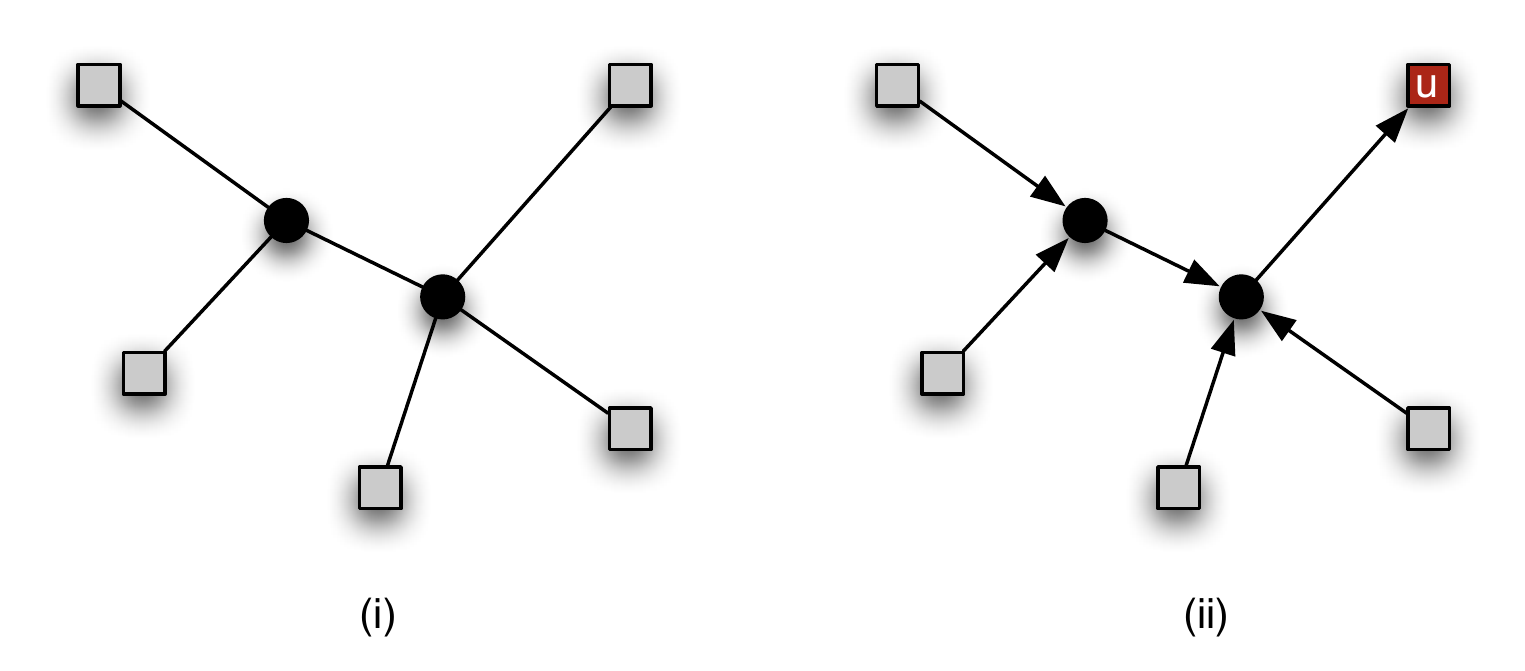}
  \caption{Example of a full components and a directed full component.}
\end{center}
\end{figure}

Figure~\ref{fig:fullcomponents}.(i) shows an example of a full component spanning a set of
terminals (squares). There are several equivalent hypergraphic LPs
\cite{CKP10}; here we focus on the {\em directed component relaxation}
(DCR) that was first introduced by Polzin and Vahdati
Daneshmand~\cite{PVd03}, and then later used by Byrka et
al.~\cite{BGRS10,BGRS11}. We will now describe this LP. 

Given a full component $K$ and one of its terminals $u$, we obtain a
{\em directed} full component by orienting all of $K$'s edges towards
the {\em sink} node $u$. Vertices $v \in K\setminus u$ are called {\em
  sources}; an illustration is given in Figure~\ref{fig:fullcomponents}.(ii). Note that when there are no Steiner-Steiner edges, i.e. when the instance is quasi-bipartite, then every full component is associated with only one Steiner vertex which we call the \textit{centre} of the full component. 

In the following, we let \K\ denote the set of all directed full
components, and for $K \in \K$, we let \sink(K)\ be the sink node of
$K$. We will sometimes abuse notation, and use $K \in \K$ for the set
of arcs of the corresponding oriented full component, and for the set
of terminals it spans interchangeably.  We use $c_K$ for the cost of the
full component $K$. For a set $U \subseteq R$, we let $\dplus(U)$
denote the set of components $K \in \K$ whose sink lies outside $U$
and that have at least one source in $U$. In this case, we will also say that $K$ \textit{crosses} $U$. We also use $x(S)$ as a
short for $\sum_{j \in S} x_j$. (DCR) has a variable for every
$K \in \K$, and a constraint for every set $U \subseteq R\setminus
r$, where $r \in R$ is an arbitrarily chosen root node. In the
following we say that $U \subseteq V$ is \textit{valid} if it contains at
least one terminal, but not the root. 

\medskip \noindent \hspace*{-1cm}
\begin{minipage}{.44\textwidth} 
  \begin{align}
    \min \quad & \sum_{K \in \K} c_K x_{K} \tag{DCR} \label{dcr} \\
    \st \quad & x(\dplus(U)) \geq 1 \quad \forall
    U \subseteq R  \notag\\
    & x \geq 0 \notag
  \end{align}
\end{minipage} 
\hspace*{.02\textwidth}
 \vline
\hspace*{.02\textwidth}
\begin{minipage}{.44\textwidth}
  \begin{align}
    \min \quad & \sum_{a \in A} c_a x_a \tag{BCR} \label{bcr} \\
    \st \quad & x(\delta^+(U)) \geq 1 \quad \forall
    \mbox{ valid } U \subseteq V  \notag\\
    & x \geq 0 \notag
  \end{align}
\end{minipage}

\bigskip

Goemans et al.~\cite{GO+11} recently showed that solving \dcr\ is
strongly NP-hard. Nevertheless, for any fixed $\epsilon$ there exist an efficient  
$(1+\epsilon)$-approximation for the value of \dcr\ in the following sense: 
Let \dcrk\ be the version of \eqref{dcr} that omits
variables for full components with more than $k$ terminals. Borchers
and Du~\cite{BD97} showed that the optimum value of \dcrk\ is larger
than that of \eqref{dcr} by at most a factor $\rho_k$ where
$$ \rho_k = \frac{(t+1)2^t + s}{t2^t+s}, $$
where we let $t\in \naturals$ and $s < 2^t$ such that $k=2^t+s$.
Byrka et al.~\cite{BGRS10,BGRS11} compute such an approximate solution
of \dcr, and the performance guarantee of their algorithm is
$\rho_k\cdot \ln 4$; for every $\epsilon$, the value $k$ can be chosen large enough such that this is at most $\ln(4)+\epsilon$. One easily sees
that already for moderately small values of $\epsilon$, large values
of $k$ need to be chosen. E.g., for $\rho_k\cdot \ln 4$ to be smaller
than $1.6$, we need $k$ to be bigger than $90$ (compare this to $\ln 4 \leq 1.39$). For such values of
$k$, solving \dcrk\ becomes a challenge since even compact
reformulations of \dcrk~\cite{BGRS10,BGRS11} have $O(n^k)$ variables,
and equally many constraints.

In this paper, we study the {\em bidirected cut relaxation} \eqref{bcr}
~\cite{E67}. In this relaxation, we convert first the original instance $G=(V,E)$ into the digraph $D=(V,A)$, where $A$ has arcs $(u,v)$ and $(v,u)$ for every edge $uv \in E$; both arcs have the same cost as $uv$. We once again
pick an arbitrary root terminal $r \in R$, and call a set $U \subseteq
V$ {\em valid} if it contains terminals but not the root. \eqref{bcr}
has a variable for every arc in $A$, and a constraint for every valid set.

Despite the fact that this relaxation is widely considered to be
strong, its integrality gap is only known to be at least
$36/31$~\cite{BGRS11}, and at most $2$. The known lower and
upper bounds on the integrality gap of \eqref{dcr}, on the other hand,
are $8/7$~\cite{KPT11} and $\ln(4)$~\cite{GO+11}.

In this note we focus on the class of {\em quasi-bipartite} Steiner tree
instances -- instances, where no two Steiner nodes are connected by an
edge. Our main result for such instances with $n$ many vertices and $m$ many edges is the following.

\begin{thm}\label{thm:main}
  For quasi-bipartite Steiner tree instances, \eqref{dcr} can be
  solved exactly using $O(mn^3)$ minimum $s,t$-cut computations in 
  graphs with $O(mn)$ vertices.
\end{thm}

We accomplish this by solving \eqref{bcr}, and by giving an efficient
\textit{decomposition algorithm} that maps the given minimal \eqref{bcr} solution
to one of \eqref{dcr}. We note that Chakrabarty et al.~\cite{CKP10}
had previously shown that \eqref{bcr} and \eqref{dcr} have the same
optimal values in quasi-bipartite graphs. The proof in \cite{CKP10}
uses ``dual'' arguments, however, and it is not clear how to obtain a
``primal'' algorithm. 

The above theorem has a couple of consequences. First, we can use it
together with \cite{BGRS10,BGRS11} to obtain an efficient
$73/60$-approximation for quasi-bipartite Steiner tree instances. We also
obtain a slightly weaker $1.28$-approximation that uses a particularly
simple sampling strategy based on \eqref{bcr}.

We remark that Goemans et al.~\cite{GO+11} have recently obtained
an alternative proof of Theorem \ref{thm:main}. The work presented
here was obtained before \cite{GO+11} appeared on the arXiv, and is
therefore independent. 

\section{Decomposing \eqref{bcr} extreme points}\label{sec:decomp}

In this section we provide a proof of Theorem \ref{thm:main}. In the
following fix a quasi-bipartite instance of the Steiner tree problem.
Let $G=(V,E)$ be the input graph, $R\subseteq V$ the set of terminals,
and $c_e$ a non-negative cost for each of the edges $e \in E$.  Also
let $D=(V,A)$ be the digraph obtained from $G$ by replacing each edge
$e=uv$ by two arcs $(u,v)$, and $(v,u)$ each having cost $c_e$. We
choose a fixed root node $r \in R$, and call a set $U \subseteq V$
{\em valid} if it contains some terminals, but not the root.

Let $y$ be a feasible solution for \eqref{dcr}. We define the
following natural map from the space $\R^{\K}$ to $\R^A$:
$$ \Phi(y) = \sum_K \chi_K \cdot y_K, $$
where $\chi_K$ is the characteristic vector of the arcs of full
component $K$.  The proof of the following observation is straight
forward, and makes use of the fact that a full component crosses a
valid set $U$ only if at least one of its arcs does. 

\begin{obs}\label{obs:phi}
  If $y$ is feasible for \eqref{dcr} then $\Phi(y)$ is feasible for
  \eqref{bcr}. 
\end{obs}

Notice that $\Phi$ is cost-preserving, and it therefore follows
immediately that the optimum solution value of \eqref{bcr} is at most
that of \eqref{dcr}. In order to prove Theorem \ref{thm:main} it
suffices to show that, in the case of quasi-bipartite graphs, the
optimum of \eqref{dcr} is at most the optimum of \eqref{bcr} as well.
We accomplish this by showing that, given 
a minimal solution $x$ of \eqref{bcr}, we can efficiently find a
minimal solution $y$ of \eqref{dcr} such that $\Phi(y)=x$. 
We start by giving an overview of the proof. We define the following
polyhedron:
\begin{equation}\label{eq:Idef}
  \I := \set{(x,y)\in\R_+^A\times\R_+^{\K}
            : x(\delta^+(U))
              + y(\Delta^+(U))
              \ge 1,
            \quad\forall \mbox{ valid } U\subseteq V}.
\end{equation}
Clearly, if $x$ is feasible for \eqref{bcr} then $(x,\0) \in \I$. Call
a full component $K \in \K$ {\em feasible} with respect to $(x,y) \in
\I$ if we can shift fractional $\lambda$-weight from the arcs of $K$ to the
full component $K$. Formally, $K$ is feasible if
\begin{equation}\label{eq:I}
  (x-\lambda \cdot \chi_K, y+\lambda \cdot e_K) \in \I, 
\end{equation}
for some $\lambda > 0$,
where $e_K$ is the standard orthonormal vector indexed by full components in $\K$.
Our first goal then is to show in Section~\ref{subsec:existential}
that a feasible component always exists. Then in
Section~\ref{sec:fc-eff1} we show how to efficiently compute such a
feasible component $K$, which allows us to find in
Section~\ref{sec:fc-eff2} the maximum $\lambda$ corresponding to $K$
such that \eqref{eq:I} holds. Our strategy then is
self-evident. Starting with the initial feasible vector
$(x^0,y^0)=(x,\0)$ to $\I$, we define a sequence of values
$\lambda^1,\lambda^2,\ldots$ as above, giving rise to a sequence of
feasible vectors $(x^i,y^i)=(x^{i-1}-\lambda^i \cdot \chi_K,
y^{y-i}+\lambda^i \cdot e_K)$ to $\I$, where $K$ is the full component
corresponding to the value $\lambda^i$. Finally, in Section~\ref{sec:
  putting-things-together} we argue that the sequence above converges
in polynomial many steps into a feasible vector $(\0,y)$ to
$\I$. Since the weight shifting at every step preserves the total
cost, our main theorem follows.

We now fill in the details, and begin with a few existential
results. Subsequently, we show how to obtain a strongly polynomial
decomposition algorithm.

\subsection{Existential results}\label{subsec:existential}

In the following it will be convenient to study slight generalizations
of \eqref{bcr} and \eqref{dcr}. Let $f$ be an {\em intersecting
  supermodular} function defined on subsets of terminals; i.e., we
have 
$$ f(A) + f(B) \leq f(A \cap B) + f(A \cup B), $$
for any $A, B \subseteq R$ with $A \cap B \neq \emptyset$. We then
obtain the LPs \eqref{bcrf} and \eqref{dcrf} by replacing $\1$ by
$f$ on the right-hand sides. 

\medskip\noindent\hspace*{-1cm}
\begin{minipage}{.44\textwidth} 
  \begin{align}
    \min \quad & \sum_{K \in \K} c_K x_{K} \tag{DCR$_{f}$} \label{dcrf} \\
    \st \quad & x(\dplus(U)) \geq f(U) \quad \forall
    U \subseteq  \notag
    R \notag \\
    & x \geq 0 \notag
  \end{align}
\end{minipage} 
\hspace*{.02\textwidth}
 \vline
\hspace*{.02\textwidth}
\begin{minipage}{.44\textwidth}
  \begin{align}
    \min \quad & \sum_{a \in A} c_a x_a \tag{BCR$_{f}$} \label{bcrf} \\
    \st \quad & x(\delta^+(U)) \geq f(U\cap R) \quad \forall
    \mbox{ valid } U \subseteq V  \notag\\
    & x \geq 0 \notag
  \end{align}
\end{minipage}
\bigskip

\noindent Chakrabarty et al.~\cite{CKP10} showed that the optimal
values of \bcr\ and \dcr\ coincide for quasi-bipartite instances. It
is an easy exercise to see that their proof extends to \bcrf\ and
\dcrf. We provide an alternate proof of this fact in the
appendix.

\begin{thm}\label{thm:ckpe}
The optimal values of \bcrf\ and \dcrf\ coincide for quasi-bipartite instances, non-negative (not necessarily symmetric) costs $c$, and intersecting supermodular function $f$. 
\end{thm}

\noindent The following is now an easy corollary.

\begin{lemma}\label{lem:decomp-exist}
  If $G$ is quasi-bipartite, $f$ is intersecting supermodular, and $x$
  is an extreme point of \bcrf, then there exists an extreme point
  $y$ of \dcrf\ such that $\Phi(y) = x$.
\end{lemma}
\begin{proof}
  By the theory of linear programming, there is $c\in\R^A$ such that
  $x$ is the unique optimal solution of \bcrf. Since the feasible
  region of \bcrf\ is upward-closed, we have $c\ge 0$, for otherwise
  $x$ would not be optimal. We claim next that \dcrf\ is
  feasible. Indeed, since \bcrf\ is feasible, we know
  $f(R),f(\emptyset)\le0$, and hence \dcr\ is feasible. Any feasible
  solution to \dcr\ can now be scaled to obtain a feasible solution to
  \dcrf.

  Since \dcrf\ is feasible and $c\ge \0$, we may let $y$ be an optimal extreme
  point solution to \dcrf. By Theorem~\ref{thm:ckpe}, $c^Ty =
  c^Tx$. Let $\tilde{x}=\Phi(y)$, and observe that 
  $c^T\tilde{x}=c^Ty=c^Tx$ since $\Phi$ preserves cost. Observation
  \ref{obs:phi} applies also to \bcrf\ and \dcrf\ and shows that 
  $\tilde{x}$ is feasible for \bcrf. As $x$ is the unique
  optimal solution to \bcrf\ for costs $c$, we must have
  $x=\tilde{x}$, and this completes the proof.
\end{proof}

The rest of this section focuses on making the above existential proof
constructive. In the following we once more abuse notation, and use
$\delta^+(U)$ ($\Delta^+(U)$) as the incidence vector of arcs (full
components) that cross valid set $U$; $\delta^+_a(U)$, and
$\Delta^+_K(U)$ then denote the component of this vector corresponding
to arc $a$ and full component $K$, respectively. We obtain the
following plausible lemma.

\begin{lemma}\label{lem:Delta-submod}
  For every $K\in\K$, $\Delta_K^+$ is submodular i.e. if
  $U,W\subseteq R$, then
$$
  \Delta_K^+(U) + \Delta_K^+(W) \ge \Delta_K^+(U\cap W) + \Delta_K^+(U\cup W).
$$
\end{lemma}
\begin{proof}
  We proceed by case analysis. If the right-hand side is zero, the
  inequality is trivial. \textit{Case 1:} Suppose $\Delta_K^+(U\cap
  W)=1$ and $\Delta_K^+(U\cup W)=0$. Then, without loss of generality,
  the sink of $K$ lies outside $U$. Since $K$ has a sink in $U\cap W$,
  in particular in $U$, this implies $\Delta_K^+(U)=1$.  \textit{Case
    2:} Suppose $\Delta_K^+(U\cap W)=0$ and $\Delta_K^+(U\cup
  W)=1$. Then, without loss of generality, $K$ has a source inside
  $U$. Since the sink of $K$ lies outside $U\cup W$, in particular
  $U$, this implies $\Delta_K^+(U)=1$. \textit{Case 3:} Finally,
  suppose $\Delta_K^+(U\cap W)=1$ and $\Delta_K^+(U\cup W)=1$. Then
  $K$ has a source inside $U\cap W$ and its sink lies outside both $U$
  and $W$, so $\Delta_K^+(U)=\Delta_K^+(W)=1$.
\end{proof}

The following lemma shows that a \bcr\ extreme point can indeed be
decomposed iteratively into full components. 

\begin{lemma}\label{lem:ex-comp}
  Let $G$ be quasi-bipartite, let $f$ be intersecting supermodular,
  let $x$ be a minimal feasible solution of \bcrf\, and let $vu\in A$
  be such that $v\notin R,u\in R$, and $x_{vu}>0$.  Then there exists
  $\lambda>0$ and $K\in\comp$ with $vu\in K$ such that
  \begin{equation}\label{eq:xprime}
    x' := x-\lambda \chi_K 
  \end{equation}
  is minimally feasible in \BCR{f'} where $f'$ is obtained from $f$ by
  reducing $f(U)$ by $\lambda$ for all valid $U$ that are crossed by
  $K$. Moreover, for any such $\lambda$, $f'$ is again intersecting
  supermodular. 
\end{lemma}
\begin{proof}
  As $x$ is a minimal feasible solution to \bcrf\ we can write it as
  $$ x := \sum_{i=1}^k \alpha_i x^i, $$ 
  where $x^1, \ldots, x^k$ are \bcrf\ extreme points, $\alpha \geq
  \0$, and $\1^T\alpha=1$.  Since $x_{vu}>0$, $x^j_{vu}>0$ for some
  $j$.  By Lemma~\ref{lem:decomp-exist}, for every $i$ there exit $y^i$ such that
  $\Phi(y^i) = x^i$ and $y^i$ is feasible to \dcrf.  Let
  $K\in\comp$ be any component with $vu\in K$ and $y^j_K>0$. 

  Clearly, $y=\sum_{i=1}^k \alpha_i y^i$ is a feasible point of
  \dcrf. Now obtain $y'$ by reducing the $K$th component of $y^j$
  in the above convex combination to $0$; i.e., let
  $$ y' = y - \alpha_j y^j_K e_K. $$
  Let $f' = f-\alpha_j y^j_K \Delta^+_K$; i.e., we obtain $f'$ from
  $f$ by reducing $f(U)$ by $\alpha_jy^j_K$ if
  $K$ crosses $U$.
  Note that $f'$ is intersecting supermodular as $f$ is intersecting
  supermodular, and $\Delta^+_K$ is submodular.

  Simply from the definition of $f'$ it now follows that $y'$ is
  feasible for \DCR{f'}, and Observation \ref{obs:phi} shows that
  $x'=\Phi(y')$ is feasible for \BCR{f'}. From the definitions of
  $\Phi$ and $y'$ it also follows that $x' = x - \alpha_j y^j_K
  \chi_K$, and we therefore choose $\lambda=\alpha_jy^j_K$ in
  \eqref{eq:xprime}.

  Finally, suppose for the sake of contradiction that $x'$ is not a
  minimal solution to \BCR{f'}; i.e., there is an arc $a$ and
  $\epsilon > 0$ such that $x''=x'-\epsilon e_a$ is feasible for
  \BCR{f'}. Then 
  \begin{eqnarray*}
    x(\delta^+(U)) - \lambda |\{ a \,:\, a \in K \cap \delta^+(U)\}| -
    \epsilon\delta^+_a(U) & = & \\ 
    x'(\delta^+(U)) - \epsilon \delta^+_a(U) & \geq & f'(U\cap R), 
  \end{eqnarray*}
  for all valid $U \subseteq V$. Thus, we have
  $$ x(\delta^+(U)) - \epsilon\delta^+_a(U) \geq f'(U\cap R) + \lambda
  |\{ a \,:\, a \in K \cap \delta^+(U)\}|, $$
  for all valid $U$. 
  Note that the right hand side of this inequality is at least $f(U
  \cap R)$ as $\Delta^+_K(U) \leq \chi_K(\delta^+(U))$ for all $U$. 
  Thus, $x$ is not a minimal \bcrf\ solution, and this is the desired
  contradiction. 
\end{proof}

\subsection{Towards efficiency I : Finding feasible components}
\label{sec:fc-eff1}

What conditions are sufficient for a full component $K$ to be feasible?
Well, certainly we need $x_a>0$ for all $a
\in K$. Beyond this, feasibility is characterized by {\em tight} valid
constraints. We say that valid set $U$ is tight for $(x,y)$ if the
corresponding constraint in \I\ is satisfied with equality.  It is
easy to see that $K$ is valid iff every tight set crossed by $K$ is
crossed by at most one of its arcs; i.e.,
\begin{equation}\label{eq:feas}
  \Delta^+_K(U) \geq \sum_{a \in K} \delta^+_a(U)
\end{equation}
holds for all tight sets $U$. 
In fact, it suffices to look at certain tight sets.

\begin{figure}
  \begin{center}
    \includegraphics[scale=0.7]{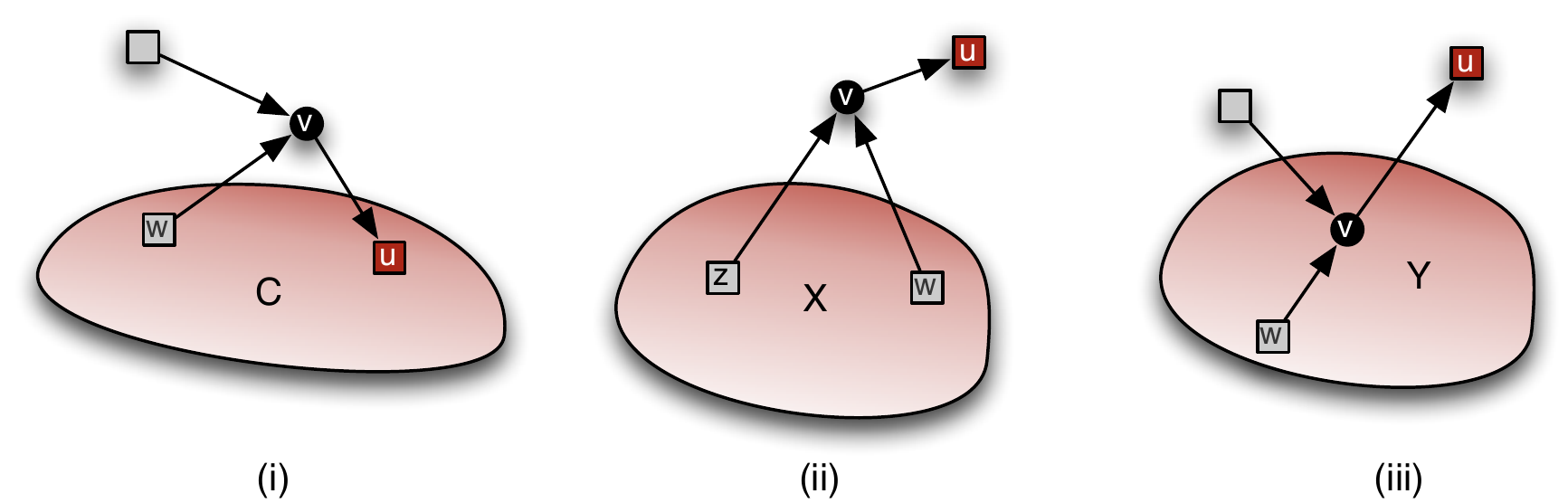}
  \end{center}
  \caption{\label{fig:feasible} Three classes of tight sets.}
\end{figure}

In what follows, fix a full component $K$ along with its centre $v$ and sink $u$.
Figure \ref{fig:feasible}.(i) shows a tight set $C$ that contains both
the sink $u$ and a source terminal $w$ but not the centre $v$. In
this case, the arc $(w,v)$ crosses $C$, but $K$ does not, and 
\eqref{eq:feas} is violated. We let \C\ be the set of neighbours of
$v$ that don't lie in such a tight set, and are hence eligible source
nodes:
$$ \C = \{ w \in \Gamma(v)\setminus u \,:\, 
     \nexists \mbox{ tight valid set } U \mbox{ with } u,w\in U
     \mbox{ and } v\notin U \}. $$
Figure \ref{fig:feasible}.(ii) shows a tight set $X$ that does not
contain the centre $v$ nor the sink $u$ of $K$, but two sources $w$
and $z$. Since two of $K$'s arcs cross $X$, component $K$ is once again not feasible. 
A feasible component may contain at most one
source from a tight set $X$ like this. We let \X\ be the set of all eligible source node sets
contained in such tight sets:
$$ \X = \{ X \cap \C \,:\, X \mbox{ is a tight valid set with } u,v
\not\in X\}. $$ 
Finally, Figure \ref{fig:feasible}.(iii) shows a set $Y$ that contains
$K$'s centre but not its sink. In this case, $K$ must contain one of
its sources in $Y$ as otherwise $K$ would not cross $Y$. We let \Y\ be
the set of all eligible source node sets contained in such tight sets:
$$ \Y = \{ Y \cap \C \,:\, Y \mbox{ is a tight valid set with } v \in
Y, u \not\in Y\}. $$

\begin{lemma}\label{lem:feas-detail}
  A component $K$ with centre $v$ and sink $u$ is feasible iff 
  (a) $x_a>0$ for all $a \in K$,
  (b) sources of $K$ lie in \C,
  (c) $K$ has at most one source in each set in \X, and
  (d) $K$ has at least one source in each set in \Y.
\end{lemma}
\begin{proof}
  If $K$ is feasible then clearly (a)-(d) above needs to be
  satisfied. We prove the converse. 

  Suppose that (a)-(d) are satisfied for some full component $K$ with
  centre $v$ and sink $u$. Since $x_a>0$ for all $a \in K$ it suffices
  to check that \eqref{eq:feas} holds for all tight valid sets
  $U$. 

  Consider a particular tight valid set $U$, and suppose first that $K$
  crosses $U$; i.e., $K$ has its sink outside $U$, and at least one of
  its sources is in $U$. Then \eqref{eq:feas} is satisfied if
  $\delta^+(U)$ has at most one of $K$'s arcs. Suppose for the sake of
  contradiction that $\delta^+(U)$ has more than one arc from $K$. In
  this case, $v \not\in U$, and $U\cap \sources(K) \in \X$. But in this case, (b)
  implies that $K$ can have at most one source in $U$; a
  contradiction. 

  Now suppose that $K$ does not cross $U$, and assume for
  contradiction that $\delta^+(U)$ has some of $K$'s arcs. Assume
  first that $(v,u) \in \delta^+(U)$. In this case, $U\cap \sources(K) \in \Y$, and
  hence, by (d), $K$ must have a source in $U$, and therefore $K$
  crosses $U$; a contradiction. Now assume that some arc $(w,v) \in K$
  crosses $U$. In this case $w$ is a source of $K$, and $u$ must be in
  $U$ as otherwise $K$ would cross $U$. But this means that $w \not\in
  \C$, and we arrive yet again at a contradiction. 

  Thus, $K$ satisfies the condition in \eqref{eq:feas} for all tight
  sets $U$. 
\end{proof}

We need the following standard uncrossing lemma. 

\begin{lemma}\label{lem:uncross}
  Let $S,T\subseteq V$ be tight such that $S\cap T\cap
  R\not=\emptyset$. Then $S\cap T$ and $S\cup T$ are also tight valid
  sets.
\end{lemma}
\begin{proof}
  Since $S\cap T\cap R\not=\emptyset$, $S\cap T$ and $S\cup T$ are
  valid, and hence 
\[
\begin{array}{rcl}
  2-y(\Delta^+(S)) -y(\Delta^+(T))   & = & x(\delta^+(S)) + x(\delta^+(T)) \\
    & \ge & x(\delta^+(S\cap T)) + x(\delta^+(S\cup T)) \\
    & \ge & 2-y(\Delta^+(S\cap T) -y(\Delta^+(S\cup T)) \\
    & \ge & 2-y(\Delta^+(S)) -y(\Delta^+(T)), \\
\end{array}
\]
where the first inequality uses the submodularity of
$x(\delta^+(\cdot))$, the second inequality follows from feasibility
of the constraints in \I, and the last inequality uses
Lemma~\ref{lem:Delta-submod}.  It follows that all inequalities above
hold with equality.
\end{proof}

The last puzzle piece needed before we present an algorithm to find
feasible components is the following structural fact.

\begin{lemma}\label{lem:disjoint}
  \X\ and \Y\ are closed under intersection and union.
\end{lemma}
\begin{proof}
  Suppose $X_1,X_2 \in \X$ and $X_1 \cap X_2 \neq \emptyset$. Then,
  for $i \in \{1,2\}$, there is a tight valid set $U_i$ that does
  not contain $v$ and $u$, and $X_i= U_i \cap C$. Clearly, $U_1 \cap
  U_2$ and $U_1 \cup U_2$ are also valid, and they are tight by
  Lemma \ref{lem:uncross}. Neither $U_1 \cap U_2$ nor $U_1 \cup U_2$
  contain $u$ and $v$. Thus $U_1\cap U_2 \cap \C = X_1 \cap X_2$ and
  $U_1 \cup U_2 \cap C = X_1 \cup X_2$ are also part of \X.

  Similarly, if distinct $Y_1,Y_2\in\Y$ intersect then, for $i \in
  \{1,2\}$, there exists a tight valid set $U_i$ with $u\notin U_i$,
  $v\in U_i$ and $Y_i=U_i\cap C$. So by Lemma~\ref{lem:uncross},
  $U'=U_1\cap U_2$ and $U''=U_1 \cup U_2$ are tight valid sets as
  well. Both $U'$ and $U''$ contain $v$ and not $u$. Therefore $U'
  \cap \C = Y_1 \cap Y_2$ and $U'' \cap \C = Y_1 \cup Y_2$ are also
  sets of \Y.
\end{proof}

Lemma \ref{lem:disjoint} allows us to slightly refine the conditions
in Lemma \ref{lem:feas-detail}. Let us define $\X^*$ to be the set of
inclusion-wise maximal sets in \X, and let $\Y^*$ be the
inclusion-wise minimal elements of \Y. Then we may replace in the
statement of Lemma~\ref{lem:feas-detail} the \X\ in (c) by $\X^*$, and
the \Y\ in (d) by $\Y^*$.  Moreover, the sets $\X^*, \Y^*$ along with
$\C$ can be found in polynomial time as we explain in the next lemma.
\begin{lemma}\label{lem:find xStar yStar C}
  For $(x,y) \in \I$, 
  the sets $\C, \X^*, \Y^*$ can be found using $O(n)$ minimum-capacity
  $s,t$-cut computations.
\end{lemma}
\begin{proof}
 Let $D$ be the digraph that has node set
  $$ V \cup \{s\} \cup \{v_K \,:\, K \mbox{ full component with }
  y_K>0\}. $$ 
  $D$ has an arc for each arc $a$ in the support of $x$;
  the capacity of this arc will be $x_a$. For each full component $K$
  in the support of $y$, we add $K$'s arcs, using node $v_K$ instead
  of $K$'s real centre $v$. 
  The sink arc $(v_K,u)$ has capacity $y_K$, and all source arcs have
  infinite capacity.  We will augment this graph and specify
  capacities in order to find \C, $\X^*$, and $\Y^*$.

  In order to determine whether $w \in \Gamma(v)\setminus u$ is
  in \C\ we need to check whether there is a tight valid set $U$ that
  contains both $u$ and $w$ but not $v$. We obtain the graph $D^{\C}_w$ from
  $D$ by adding arcs $(s,w)$ and $(s,u)$ of infinite capacity. We also 
  assign infinite capacity to arc $(v,r)$.
  Using the feasibility of $(x,y)$ for \I\ it follows that
  any cut separating $s$ and $r$ has capacity at least
  $1$. Furthermore, the minimum-capacity such cut has capacity $1$ if
  and only if $w \not\in \C$. In order to compute \C\ it suffices to
  check all $w \in \Gamma(v)\setminus u$. 

  The strategy to find $\X^*$ is very similar to the above procedure
  for determining \C. For each terminal $w \in \C$ we find, if it
  exists, a inclusion-wise maximal valid set $U$ containing $w$ but
  not $u$ and $v$. In order to do this, we obtain graph $D^\X_w$ from
  $D$ by adding arc $(s,w)$ of infinite capacity, assign infinite
  capacity to arcs $(v,u)$, and $(u,r)$. Once again, the minimum
  $s,r$-cut in this graph has capacity at least $1$, and it is exactly
  $1$ if \X\ has a set containing $w$. In the latter case, it suffices
  to compute a maximal min $s,r$-cut in this graph, and include it in
  $\X^*$. After having done this for all $w \in \Gamma(v)\setminus u$,
  and after deleting all non-maximal sets, Lemma \ref{lem:disjoint}
  implies that $\X^*$ is a family of pair-wise disjoint sets.

  Finally, in order to compute $\Y^*$, we create the following graph
  $D^\Y_w$ for every $w \in \C$: add two arcs $(s,w)$, and $(s,v)$ of
  infinite capacity to $D$, and assign infinite capacity to
  $(u,r)$. Once more by feasibility, a maximum $s,r$-flow in this
  graph has value at least $1$, and value exactly $1$ if there is a
  \Y-set containing $w$. In the latter case, we compute an
  inclusion-wise minimal mincut and add its intersection with \C\ to
  $\Y^*$. We repeat the procedure for all $w \in \C$. By Lemma
  \ref{lem:disjoint}, the family $\Y^*$ contains pair-wise
  disjoint sets, once we clean up by deleting all non-minimal sets.

Finally, note that in all cases above, we perform $n$ many mincut computations. 
\end{proof}

We are now ready to show how to efficiently find a feasible component. 
\begin{lemma}\label{lem:find-fc}
  Let $(x,y) \in \I$, and suppose that there is a feasible
  component. Then there is an algorithm to find such a component that
  runs in time $O(n\tau_{mc})$, where $\tau_{mc}$ is the time needed
  to find a minimum-capacity $s,t$-cut. 
\end{lemma}
\begin{proof}
  Choose a Steiner vertex $v$, and sink node $u$ such that
  $x_{vu}>0$. We know from Lemma \ref{lem:ex-comp} that there is a
  feasible component $K$ with centre $v$ and sink $u$. By
  Lemma~\ref{lem:find xStar yStar C}, the corresponding sets $\C,
  \X^*, \Y^*$ can be computed in time $O(n\tau_{mc})$. We can then find
  a feasible component with centre $v$ and sink $u$ by computing a max
  flow in a bipartite auxiliary graph. Introduce a vertex $x$ for
  every set $X \in \X^*$, and a vertex $y$ for every $Y \in \Y^*$. Add
  an arc $(x,y)$ if the corresponding sets $X$ and $Y$ share a
  terminal from \C. Also connect each of the $\Y^*$ nodes to a sink
  node $t$, and give each of these arcs unit capacity. Similarly,
  introduce a source node $s$, and connect it to all $\X^*$ nodes via
  unit-capacity arcs. Observe that a maxflow of value $|\Y^*|$ exists
  iff there is a feasible component with sink arc $(v,u)$. Let $h$ be
  such a maximum flow, and let $S$ be the set of terminals
  corresponding to edges $(x,y)$ with $h_{xy}=1$. It follows from
  Lemma \ref{lem:feas-detail} that
  $$ \{ (w,v) \,:\, w \in S \} \cup \{(v,u)\} $$
  is a feasible full component.
\end{proof}

\subsection{Towards efficiency II : Finding the step weight $\lambda$}
\label{sec:fc-eff2}

In this section we assume that we have a minimal feasible point $(x,y)
\in I$, and a feasible component $K$. The following lemma establishes
that we can find the largest $\lambda$ such that 
$$ (x^\lambda, y^\lambda) := (x - \lambda \chi_K, y + \lambda e_K) $$
is in \I. 

\begin{lemma}\label{lem:lambda}
  Given a minimal feasible point $(x,y) \in \I$, we can find the largest
  $\lambda$ such that $(x^\lambda, y^\lambda)$ is feasible for \I. Our
  algorithm runs in time $O(n^2 \tau_{mc})$. 
\end{lemma}
\begin{proof}
  Let us first choose $\lambda^0 = \min_{a \in K} x_a$; clearly, 
  a larger value of $\lambda$ would result in some negative $x$
  variables.  $(x^{\lambda^0}, y^{\lambda^0})$ may still not be
  feasible, and violate some of the valid cut inequalities. We now
  look for a valid set $U$ that is violated the most. 

  Once again this is accomplished by min $s,r$-cut computations in a
  suitable auxiliary graph. Do the following for each $w \in R$. Start
  with the graph $D$ used in Lemma \ref{lem:find-fc}. Let the capacity
  of every arc $a \in A$ be $x^{\lambda^0}_a$, and let the capacity of
  arc $v_k,u$ be $y^{\lambda_0}_K$ for all $K \in \K$. Finally add an
  arc $(s,w)$ of infinite capacity. If $(x^{\lambda^0}, y^{\lambda^0})$
  is feasible then the max $s,r$-flow in this graph is at least
  $1$. If it is lower, let $U_w$ be the vertex set corresponding to a
  minimum $s,r$-cut. 

  Among all the sets $U_w$ found this way, let $U^0$ be one of minimum
  capacity. Choose $\lambda^1 < \lambda^0$ such that $(x^{\lambda^1},
  y^{\lambda^1})$ satisfies the cut constraint for set $U^0$. The new
  point $(x^{\lambda^1}, y^{\lambda^1})$ may still not be
  feasible. There may be a valid set $U$ that is violated by this
  point. As a function of $\lambda$, the violation of the constraint
  for set $U$ is
  $$ h_U(\lambda) = (1-x(\delta^+(U)) - y(\Delta^+(U))) - \lambda   (|\delta^+(U) \cap K| - \Delta^+_K(U)), $$ 
  where, we recall, $|\delta^+(U) \cap K|$ is the number of arcs in
  $K$ that cross $U$, and $\Delta^+_K(U)$ is $1$ if $K$ crosses $U$,
  and $0$ otherwise.  Call the coefficient of $\lambda$ in the above
  expression $\alpha(U)$, and note that it is an integer.

  Recall now that we chose $U^0$ as the valid set with maximum
  violation. The fact that $U^0$ is not violated by $\lambda^1$, but
  $U$ means that $\alpha(U) < \alpha(U^0)$. In fact, all valid sets 
  $U'$ with $\alpha(U') \geq \alpha(U^0)$ are satisfied by
  $(x^{\lambda^1}, y^{\lambda^1})$, following the previous argument. 
  
  Note that $\alpha(U)$ is at most $n$, and non-negative. We continue
  in the same fashion: for $(x^{\lambda^1}, y^{\lambda^1})$ we look
  for a valid set $U^1$ that is maximally violated, and choose
  $\lambda^2 < \lambda^1$ largest so that this set is satisfied.
  
  This produces a sequence of $\lambda$'s and corresponding valid sets 
  $$ U^0, U^1, U^2, \ldots, $$
  such that $\alpha(U^0) > \alpha(U^1) > \alpha(U^2) >
  \ldots$. Clearly, this process has to terminate within in $n$ steps.
\end{proof}

\subsection{Efficiency: Putting things together}
\label{sec: putting-things-together}

We are now ready to state the entire polynomial-time algorithm for
computing the decomposition of a minimal \bcr\ solution $x$. 

\begin{algorithm}{\label{alg:deco}Decompose}
\begin{algorithmic}[1]
\REQUIRE $x$ is a minimal feasible solution of \bcr.
\STATE Initialize $y\in\R_+^\comp$ to $0$.
\WHILE {there is $vu\in A$ with $v\notin R,u\in R,x_{vu}>0$}
  \WHILE {$x_{vu}>0$}
	\STATE Find a feasible component $K\in\comp$,
			with centre $v$ and sink $u$. \hfill
                        (Lemma \ref{lem:find-fc}) \label{alg:deco:1}
	\STATE Find the greatest $\lambda> 0$ such that
			$(x^{\lambda},y^\lambda)\in I$. \hfill 
                        (Lemma \ref{lem:lambda}) \label{alg:deco:2}
	\STATE Set $(x,y)$ to $(x^{\lambda},y^\lambda)\in I$. \label{alg:deco:3}
	\ENDWHILE
\ENDWHILE
\RETURN $y$.
\end{algorithmic}
\end{algorithm}

Our algorithm maintains as an invariant that $(x,y)$ is a minimal
feasible point in \I. Note that Lemma \ref{lem:Delta-submod} implies
that the function $f$ defined by  
$$ f(U) = 1 - y(\Delta^+(U)) $$
for all valid $U \subseteq R$ is intersecting supermodular. Lemma
\ref{lem:ex-comp} then guarantees the existence of a feasible
component $K \in \K$. Lemma \ref{lem:lambda} implies that the above
invariant is maintained throughout. It remains to show that steps
\ref{alg:deco:1} -- \ref{alg:deco:3} are executed a polynomial number
of times. 

Call a step {\em saturating} if the support $x$ decreases; i.e., some
arc variable $x_a$ is decreased to $0$. Obviously, the number of such
events are upper bounded by $O(m)$, where $m$ is the number of edges
in the original graph $G$. 

Let us focus on {\em non-saturating} steps. Let $(x,y)$ be the point
in \I\ in step \ref{alg:deco:1}, and let $K$ be the full component
chosen. We find $\lambda$ in step \ref{alg:deco:2} and note that 
the supports of $x$ and $x^{\lambda}$ have the same size. 
The increase of $\lambda$ is thus determined by some valid set
$U$ as follows: $U$ is non-tight for $(x,y)$ and tight for
$(x^\lambda,y^\lambda)$. 

Our choice of $K$ implies (see also Lemma \ref{lem:find-fc}) that sets
$U$ that are tight for $(x,y)$ are also tight for
$(x^\lambda,y^\lambda)$. $K$ is certainly not feasible for
$(x^\lambda,y^\lambda)$, and hence, once again by Lemma
\ref{lem:find-fc}, at least one of \C, $\X^*$, or $\Y^*$ must have changed.

As a set $U$ that is tight for $(x,y)$ is tight also for
$(x^\lambda,y^\lambda)$, \C\ can only shrink, and the number of times
this can happen is clearly bounded by $n$. Similarly, the new sets \X,
and \Y\ are supersets of their old counterparts.

Focus on \X, and let $\L$ and $\L'$ be maximal
laminar families in \X\ for $(x,y)$ and $(x^\lambda,y^\lambda)$,
respectively. The set $\X^*$ precisely consists of the
maximal sets of \L. If $\X^*$ changes then this means that the set of
maximal sets in laminar families \L\ and $\L'$ differ. This can happen
only for one of two reasons: the sets in $\L'$ cover more terminals than
those in \L, or two maximal sets in \L\ are now part of the same
maximal set in $\L'$. Clearly, the number of such events is bounded by
$O(|K|)=O(n)$. 

The argument for \Y\ is similar, and we omit it here. In summary we
have proved the following:

\begin{lemma}
  Between any two saturating steps, the algorithm performs at most
  $O(n)$ non-saturating ones. Thus the total number of times steps
  \ref{alg:deco:1} -- \ref{alg:deco:3} of Algorithm \ref{alg:deco} are
  executed is bounded by $O(mn)$. 
\end{lemma}

Note that this means that at most $O(mn)$ full components are added
throughout the algorithm, and that the auxiliary graph used in the
mincut computations in Algorithm \ref{alg:deco} has at most $O(mn)$
nodes. This proves Theorem \ref{thm:main}. 

\section{An application: Sampling without decomposition}

In this section, we employ the existential result given in Lemma
\ref{lem:decomp-exist} to give a compact and fast implementation of a
recent (DCR)-based LP-rounding algorithm (henceforth referred to by CKP)
given by Chakrabarty et al.~\cite{CKP2010} for the case of
quasi-bipartite Steiner tree instances.

We first review the algorithm CKP in the special case of
quasi-bipartite Steiner tree instances. Given such an instance, CKP
first solves \eqref{dcr}; let $y$ be the corresponding basic optimal
solution, and let $M=\1^Ty$. The algorithm now repeats the following
sampling step $M\ln3$ times: sample component $K \in \K$ independently
with probability $y_K/M$. In $G$, contract $K$'s cheapest edge (the
so-called {\em loss} of $K$), and continue. Let $G'$ be the final
contracted graph, and let $S$ be the set of centre vertices of the
$M\ln 3$ sampled full components. The algorithm now returns a
minimum-cost tree spanning the terminals $R$, and the set $S$. 

Chakrabarty et al. showed that the expected cost of the returned
solution is no more than $1.28$ times the value of the initial
\eqref{dcr} solution. We observe here that when it comes to
quasi-bipartite graphs, the above process that iteratively samples
components, can be alternatively interpreted as sampling their
centers. Each Steiner vertex ends up in set $S$ with a certain
probability, and this distribution can be realized alternatively
by sampling directly from a \bcr\ solution.

\begin{lemma}
\label{lem:same_distrib}
Let $y$ be a solution to \dcr\ for a given quasi-bipartite Steiner
tree instance, and let $x=\Phi(y)$. Then in any iteration of CKP, the
probability of choosing a component with center $v$ is exactly
$x\left( \delta^+(v) \right)/M$.
\end{lemma}
\begin{proof}
  Consider a Steiner vertex $v$, and let $\K_v$ be the set of full
  components that have $v$ as their centre. The definition of $\Phi$
  immediately shows that
  $$ x\left( \delta^+(v)\right) = \sum_{K \in K_v}
  y_K. $$ 
  This obviously implies the lemma as the right-hand side of the above
  equality, scaled by $M$, is the probability that a component with
  centre $v$ is sampled.
\end{proof}

Consequently, we also have $M=\1^Ty=\sum_{v \in V\setminus R}
x(\delta^+(v))$.  We can now {\em simulate}
Algorithm CKP using the optimal solution of \bcr. 

\begin{algorithm}{\label{alg:CKP2}}{CKP2}
\begin{algorithmic}[1]
\label{alg_RLC}
\REQUIRE $x$ is an optimal basic feasible solution of \bcr, and $M=\sum_{v\in V\setminus R} x \left(\delta^+(v) \right)$.
\FOR {$i = 1 \to M \ln 3$}
\STATE Sample  a Steiner vertex $v$ with probability $\frac{1}{M}x\left( \delta^+(v) \right)$.
\ENDFOR
\RETURN	a minimum spanning tree on the terminals and the sampled Steiner vertices. 
\end{algorithmic}
\end{algorithm}

To complete our argument, if we run Algorithm~\ref{alg:CKP2} on a
minimal \bcr\ solution $x$, then by Lemma~\ref{lem:same_distrib} the
expected cost agrees with that of Algorithm CKP run on the
\dcr\ solution $y$ that is obtained by decomposing $x$. Both $x$ and
$y$ are optimal for \bcr\ and \dcr\ respectively, so our claim
follows.

\bibliography{qbip}
\bibliographystyle{plain}

\section*{Appendix}

\subsection*{Proof of Theorem~\ref{thm:ckpe}}
First we need to consider the dual \ddcr{f}\ of \dcrf. To slightly simplify notation, we will denote $f(U)$ also by $f_U$, so that \ddcr{f}\ reads as follows.
  \begin{align}
    \max \quad & f^T z \tag{DCR$^D_{f}$} \label{ddcrf} \\
    \st \quad & z\left( \Delta^+_K \right) \leq c_K \quad \forall K \in \K \notag\\
    & z_U \geq 0 \quad \forall   \mbox{ valid } U \subseteq R \notag
  \end{align}

\noindent To prove the theorem, consider  $z$ to be feasible to \ddcr{f} with non-negative costs $c$. What we show next is that 
\begin{equation}\label{equa: non-trivial ineq dual BCR DCR}
f^Tz \le c^Tx
\end{equation}
is valid for \bcrf. We point here that we may assume, without loss of 
generality, that there are no arcs between terminals; this may be 
accomplished by splitting such arcs into two, putting a non-terminal in between. 
Since we are in the quasi-bipartite case, there are no arcs between non-terminals 
either. 

Before we proceed with the proof of~\eqref{equa: non-trivial ineq dual BCR DCR}, we 
need to show that we may also assume that the solution $z$ of \ddcr{f} enjoys a 
nice structural property.

\begin{lemma} \label{lem:dcr-dual-laminar}
Every optimal solution of \ddcr{f} has laminar support.
\end{lemma}

\begin{proof}
Our proof is by contradiction. Let $z$ be an optimal solution to \ddcr{f} that 
maximizes $\sum_{U\subseteq R} |U|^2z_U$. If the support of $z$ is not laminar, 
there must exist two intersecting subsets $S,T$ of the terminals with $z_S,z_T>0$. So let 
$\epsilon=\min\set{z_S,z_T}$ and define
$$
  z' := z + \epsilon\grp{e_{S\cap T} + e_{S\cup T} - e_S - e_T},
$$
and note that $z'\geq 0$. Now we claim that $z'$ is an optimal solution to \ddcr{f}. Indeed, for each $K\in\K$, submodularity of $\Delta^+_K$ and the fact that $z$ is feasible in \ddcr{f} imply
$$
z'\left( \Delta^+_K \right)
= z\left( \Delta^+_K \right) 
      + \epsilon\grp{ \Delta^+_K(S\cap T)  + \Delta^+_K(S\cup T) 
                     - \Delta^+_K(S) - \Delta^+_K(T)}
	\leq f_U.
$$
Thus, $z'$ is feasible in \ddcr{f}. Moreover, intersecting supermodularity of $f$ 
implies
$$
  f^Tz' = f^Tz + \epsilon\grp{f_{S\cap T}+b_{S\cup T}-f_S-f_T}\ge f^Tz,
$$
which proves optimality for $z'$. Note then that $x\mapsto x^2$ is strictly convex 
and
$|S\cap T|<\min\{|S|,|T|\}$. The contradiction then (assuming the non-laminarity of 
$z$) is that $$
\sum_{U\subseteq R} |U|^2z'_U 
	=  \sum_{U\subseteq R} |U|^2 z_U
          + \epsilon\grp{|S\cap T|^2 + |S\cup T|^2 - |S|^2 - |T|^2} 
	>  \sum_{U\subseteq R} |U|^2 z_U.
$$
\end{proof}

We are now ready to start the proof of~\eqref{equa: non-trivial ineq dual BCR DCR}. 
Our argument uses induction on $|\supp(z)|$.

The base case of our induction is simple since if $|\supp(z)|=0$, we have $z=0$. 
But  $x\ge0$ is valid for \bcrf and so is $c^Tx\ge0=f^Tz$.

Now suppose $|supp(z)|\ge1$. Let $T_1,\ldots,T_k$ be the inclusion-wise maximal 
sets in $\supp(z)$, as they follow from Lemma~\ref{lem:dcr-dual-laminar}. Since $
\supp(z)$ is laminar, $T_1,\ldots,T_k$ are disjoint. Next we distinguish the cases 
$k=1$ and $k\geq 2$, and for each of them (building on the inductive argument) we 
conclude that $f^Tz \le c^Tx$ is valid. In both cases below we denote by $N=V\setminus R$ 
the set of non-terminals.

\paragraph{(The case $k=1$): } The laminar family has one element $T=T_1$.
For each $v\in N$ we define 
$$
  \tau_v := \min\set{c_{vu} : vu\in A, u\in R, u\notin T}\cup\set{+\infty}.
$$
Next, order the elements of $N$ as $v_1,\ldots,v_\ell$ such that
$\tau_{v_1} \le \tau_{v_2} \le\cdots\le \tau_{v_\ell}$. Let $t_0=0$ and $t_{\ell
+1}=z_T$,
and for each $1\le i\le\ell$, let $t_i=\min\set{\tau_{v_i},z_T}$. For each $1\le i
\le\ell+1$, we also define
\begin{eqnarray*}
  z^i &:=& (t_i-t_{i-1})e_T \\
  c^i &:=& (t_i-t_{i-1})\chi_{\delta^+(T\cup\set{v_i,\ldots,v_\ell})}.
\end{eqnarray*}
Our next claim is that $f^Tz^i \le (c^i)^Tx$ is valid for \bcrf. Indeed, note that 
the inequality at hand is just a scaling (by $t_i-t_{i-1}$) of the \bcrf\ 
inequality
$$
  f_T\le \sum_{uv\in A} \delta^+_{uv}(T\cup\set{v_i,\ldots,v_\ell})x_{uv}
$$
and thus $f^Tz^i \le (c^i)^Tx$ must be valid for \bcrf. Next we define 
$$
 z' := z - \sum_{i=1}^{\ell+1}z^i , ~~~\textrm{and}~~~   c' := c - \sum_{i=1}^
{\ell+1}c^i.
$$
and we note that 
$$
 c'_{v_ju} =
 \left\{
\begin{array}{ll}
c_{v_ju} - t_j &, ~ \textrm{if}~ u\in R\setminus T \\
c_{uv_j} - z_T + t_j &, ~ \textrm{if}~ u\in T 
\end{array}
\right.
$$
since $\sum_{i=1}^{\ell+1}c^i_{v_ju} = t_j$ when $u\in R\setminus T$, while $\sum_
{i=1}^{\ell+1}c^i_{uv_j} = z_T - t_j$ when $u\in T$. 

The first important observation then is that by the definition of $z'$ we have
$$ z
     =  z' + \sum_{i=1}^{\ell+1}z^i 
     =  z' + \sum_{i=1}^{\ell+1}(t_i - t_{i-1})e_T 
     =  z' + (t_{\ell+1} - t_0)e_T 
     =  z' + z_Te_T. 
$$
This means that $|\supp(z')|=|\supp(z)|-1$. In what follows we show that (i) $c'$ 
can be thought as a non-negative cost function (see Claim~\ref{claim: c' non 
negative}) and  (ii) that $z'$ is feasible to \ddcr{f} with cost $c'$ (see Claim~
\ref{claim: z' feasible for dcr with c'}). Note that (i),(ii), along with the 
observation that $|\supp(z')|=|\supp(z)|-1$ show that $f^Tz' \le (c')^Tx$ is valid 
for \bcrf\ (due to the inductive hypothesis). This allows us to conclude that also 
the inequality
$$  b^T(z' + z^1 + \cdots + z^{\ell+1}) \le (c' + c^1 + \cdots + c^{\ell+1})^Tx,$$
is valid for \bcrf\ . The latter inequality is just $f^Tz \le c^Tx$, which completes 
the inductive argument, and the case $k=1$.

Thus it remains to argue formally about (i),(ii) above.

\begin{claim}\label{claim: c' non negative}
$c'$ as defined above is non-negative.
\end{claim}

\begin{proof}
First take $vu\in A$ with $v\in N$ and $u\in R$. If $u\in T$, then $
  c'_{vu} = c_{vu} \ge 0$. Otherwise, let $1\le j\le\ell$ be such that $v_j=v$.
But then
$$
  c'_{vu}
     = c_{vu} - t_j 
    =  c_{vu} - \min\set{\tau_v, z_T} 
     \ge  c_{vu} - \tau_v 
     \ge  c_{vu} - c_{vu} 
     =  0.
$$
For the other case, take $uv\in A$ with $u\in R$ and $v\in N$. If $u\notin T$, then
$  c'_{uv} = c_{uv} \ge 0$. Otherwise, let $1\le j\le\ell$ be such that $v_j=v$.

If $t_j=z_T$, then $c'_{uv} = c_{uv} - z_T + t_j \ge 0$. On the other hand,
if $t_j=\tau_v$, then let $w\in R\setminus T$ be such that $vw\in A$ and $
\tau_v=c_{vw}$.
Then the component $K$ with source $u$, sink $w$, and non-terminal $w$
crosses $T$, so by feasibility of $z$, we have
$$
  c_{uv} + c_{vw} \ge \sum_{U\subseteq R}\Delta^+_K(U)z_U\ge z_T,
$$
which implies that $  c'_{uv} = c_{uv} + c_{vw} - z_T \ge 0$.
\end{proof}

Finally, feasibility of $z'$ in the dual of \dcrf\ with cost function $c'$ is given 
by the next claim. 
\begin{claim}\label{claim: z' feasible for dcr with c'}
$z'$ is feasible to \ddcr{f} with costs $c'$.
\end{claim}

\begin{proof}
First, note that $z'\ge0$ follows from its definition. Now take a component $K$ 
with non-terminal $v\in V\setminus R$ and sink $u\in R$. Let $1\le j\le\ell$ be 
such that $v_j=v$.

Suppose first that $K$ does not cross $T$.
If $\sources(K)\cap T=\emptyset$,
then $\Delta^+_K(U)=0$ whenever $U\subseteq R$ and $z_U>0$. Therefore, since by 
Claim~\ref{claim: c' non negative} we have $c'\ge0$, we conclude that
$$
  \sum_{U\subseteq R}\Delta^+_K(U)z'_U = 0 \le c'(K).
$$

On the other hand, suppose $\sources(K)\cap T\not=\emptyset$ but $u\in T$.
If $t_j=z_T$, we have
$$
  \sum_{i=1}^{\ell+1} c^i(K)
    = \sum_{i=j+1}^{\ell+1}(t_i-t_{i-1})|\sources(K)\cap T| = 0.
$$
Now, since $K$ does not cross $T$, the above implies that 
$$
  c'(K) = c(K) \ge \sum_{U\subseteq R}\Delta^+_K(U)z_U
    = \sum_{U\subseteq R}\Delta_K(U)z'_U.
$$

If $t_j=\tau_v$, let $u'\in R\setminus T$ be such that $\tau_v=c_{vu'}$
and let $W=\sources(K)\cap T$. By summing the inequalities of \ddcr{f} 
corresponding
to components of the form $wvu'$ for each $w\in W$, we have
\begin{eqnarray*}
  \sum_{w\in W}(c_{wv} + c_{vu'})
    & \ge & \sum_{w\in W}(\sum_{U\subseteq R}\delta^+_{wu'}(U)z_U) \\
    & = & \sum_{U\subseteq R}(\sum_{w\in W}\delta^+_{wu'}(U))z_U \\
    & \ge & \sum_{U\subseteq R}(\sum_{w\in W}\delta^+_{wu}(U))z_U + |W|z_T \\
    & \ge & \sum_{U\subseteq R}\Delta^+_K(U)z_U + |W|z_T.
\end{eqnarray*}
Therefore
\begin{eqnarray*}
  c'(K)
    & \ge & \sum_{w\in W} c'_{wv} 
     =  \sum_{w\in W} (c_{wv} - z_T + t_j) 
     =  \sum_{w\in W} c_{wv} - |W|(z_T - t_j) \\
    & \ge & \sum_{U\subseteq R}\Delta^+_K(U)z_U - |W|c_{vu'} + |W|z_T
            - |W|(z_T - t_j) 
     =  \sum_{U\subseteq R}\Delta^+_K(U)z_U,
\end{eqnarray*}
as required.

Now suppose $K$ crosses $T$.
Let $W=\sources(K)\cap T$,
and let $K'$ be the sub-component of $K$ having sources $W$.
We have
$$  c'(K)
     \ge  c'(K') 
     =  \sum_{w\in W}c'_{wv} + c'_{vu} 
     =  \sum_{w\in W}(c_{wv} - z_T + t_j) + c_{vu} - t_j.
$$
When $t_j=z_T$, we recall that $K$ crosses $T$, and that $T$ is outermost, hence
$$
  c'(K)
     \ge  c(K') - z_T 
     \ge  \sum_{U\subseteq R}\Delta^+_{K'}(U)z_U - z_T 
     =  \sum_{U\subseteq R} \Delta^+_K(U)z'_U.
$$
In the final case where $t_j=\tau_v$,
let $u'\in R\setminus T$ achieve the maximum in the definition of $\tau_v$,
and let $W=\sources(K)\cap T$.
Then, again $c'$ is non-negative because
\begin{eqnarray*}
  c'(K)
    & \ge & \sum_{w\in W}(c_{wv} - z_T + c_{vu'}) + c_{vu} - c_{vu'} 
     \ge  \sum_{w\in W}(c_{wv} + c_{vu'}) - |W|z_T \\
    & \ge & \sum_{w\in W}(\sum_{U\subseteq R}\delta^+_{wu'}(U)z_U) - |W|z_T 
     =  \sum_{U\subseteq R}(\sum_{w\in W}\delta^+_{wu'}(U))z'_U \\
     &\ge&  \sum_{U\subseteq R}\Delta^+_{K'}(U)z'_U 
     =  \sum_{U\subseteq R}\Delta^+_K(U)z'_U. 
\end{eqnarray*}
\end{proof}

\paragraph{(The case $k\geq 2$): } Recall that by $N=V\setminus R$ we denote the 
set of non-terminals. As in the previous case, we define non-negative cost function 
$c'$ along with $z'$ of smaller support so as to use the inductive hypothesis. 
For each $1\le i\le k$, we define now $c^i\in\R_+^A$ as follows. 

For each $uv\in A$ with $u\in R$ and $v\in N$, we set
$$        c^i_{uv} := 
          \begin{cases}
            c_{uv} & \text{ if $u\in T_i$} \\
            0 & \text{ otherwise},
          \end{cases}
$$
while for each $vu\in A$ with $v\in N$ and $u\in R$ we define
$$
      c^i_{vu} := \max
        \begin{Bmatrix}
          \sum(z_U : U\subseteq T_i, U\cap W\neq\emptyset, u\notin U)
                 - \sum(c_{wv} : w\in W) \\
        : W\subseteq T_i\text{ such that }
                 wv\in A\;\forall w\in W
        \end{Bmatrix}.
$$

Our first claim is that $c^i\ge 0$, for every $1\le i\le k$. The reason is that if 
$uv\in A$ with $u\in R$ and $v\in N$, we clearly have $c^i_{uv}\ge 0$. If on the 
other hand  $vu\in A$ with $v\in N$ and $u\in R$, just take $W=\emptyset$ in the 
definition of $c^i_{vu}$ to get $c^i_{vu}\ge0$.

Our next claim is that $c$ dominates the sum of $c^i$'s, namely
\begin{equation}\label{equa: c dominates sum of ci's}
c^1 + \cdots + c^k \le c.
\end{equation}

To see why, let $uv\in A$ with $u\in R$ and $v\in N$. Since $T_i$'s are disjoint, 
we have that $c^1_{uv} + \cdots + c^k_{uv} \le c_{uv}$. Next take $vu\in A$ with 
$v\in N$ and $u\in R$. For each $1\le i\le k$, let $W_i\subseteq T_i$ achieve the 
maximum in the definition of $c^i_{vu}$. Consider the component $K$ with sources 
$W:=W_1\cup\cdots\cup W_k$,
non-terminal $v$, and sink $u$. Since $z$ is feasible to \ddcr{f} with costs $c$,
we have
\begin{eqnarray*}
  c(K)
     &\ge&  \sum_{U:~U\subseteq R, \Delta^+_K(U)=1} z_U 
     =  \sum_{i=1}^k \sum_{U:~U\subseteq T_i, \Delta^+_K(U) = 1} z_U  
     =  \sum_{i=1}^k \sum_{U:~ U\subseteq T_i, U\cap W_i\neq\emptyset, u\notin U}
                                         z_U  \\
    & = & \sum_{i=1}^k (c^i_{vu} + \sum_{w\in W_i} c_{wv}) 
     =  \sum_{i=1}^k c^i_{vu} + \sum_{w\in W} c_{wv} 
     =  \sum_{i=1}^k c^i_{vu} + c(K) - c_{vu}.
\end{eqnarray*}
The latter implies that $$  \sum_{i=1}^k c^i_{vu} \le c_{vu}$$ as we claimed. 

Next, similarly to the case $k=1$, we define $z^i\in\R_+^{2^R}$ for each $1\le i\le 
k$ as follows. For $U\subseteq R$ we set
$$
  z^i_U :=
    \begin{cases}
      z_U & \text{ if $U\subseteq T_i$} \\
      0 & \text{ otherwise}.
    \end{cases}
$$
From our definition, it is immediate that $  \sum_{i=1}^k z^i = z$. Analogously to 
Claim~\ref{claim: z' feasible for dcr with c'}, we show again that 

\begin{claim}\label{claim: z' feasible for dcr with c' in case k>=2}
For $1\le i\le k$, vector $z^i$ is feasible to \ddcr{f} with costs $c^i$
\end{claim}

\begin{proof}
Indeed, consider a component $K$ with non-terminal $v\in N$ and sink $u\in R$. Then, by 
setting $W=\sources(K)\cap T_i$ in the definition of $c^i_{vu}$, we obtain
\begin{eqnarray*}
  c^i(K)
    & = & c^i_{vu} + \sum_{w\in\sources(K)}c^i_{wv} 
     =  c^i_{vu} + \sum_{w\in W} c_{wv} \\
     &\ge&  \sum_{U:~ U\subseteq T_i, U\cap W\neq\emptyset, u\notin U}z_U  
     =  \sum_{U:~ U\subseteq R, U\cap\sources(K)\neq\emptyset,
                       u\notin U} z^i_U   \\
     &=&  \sum_{U\subseteq R} \Delta^+_K(U) z_U.
\end{eqnarray*}
Since $z^i\ge0$, this shows $z^i$ is feasible to \ddcr{f} with costs $c^i$.
\end{proof}

Since we are dealing with the case $k\ge2$, we have $|\supp(z^i)|<|\supp(z)|$, for 
all $ 1\le i\le k$, so by induction, $ b^Tz^i \le (c^i)^Tx$ is valid for \bcrf\
for each $1\le i\le k$.

By summing over all $i$'s, we get that
$$
  f^Tz = f^T(z^1+\cdots+z^k) \le (c^1+\cdots+c^k)^Tx \le c^Tx
$$
is also valid for \bcrf\, which completes the inductive proof. Altogether, this 
justifies~\eqref{equa: non-trivial ineq dual BCR DCR}.

\end{document}